\documentclass[runningheads,a4paper]{llncs}

\usepackage{amssymb}
\usepackage{amsmath,graphicx}
\usepackage{amsfonts}
\usepackage{bm}
\usepackage{algorithm,algorithmic}
\usepackage[francais,serbian,english]{babel}
\usepackage{graphicx}
\usepackage{multicol}

\usepackage{url}
\usepackage{xifthen}
\urldef{\mailsa}\path|srdjan.kitic@inria.fr|
\urldef{\mailsb}\path|nancy.bertin@irisa.fr|
\urldef{\mailsc}\path|remi.gribonval@inria.fr|
\newcommand{\keywords}[1]{\par\addvspace\baselineskip
\noindent\keywordname\enspace\ignorespaces#1}

\graphicspath{{Figures/}}

\newcommand{\Rset}[2]{\ifthenelse{\equal{#2}{1}}{\in \mathbb{R}^{\mathrm{#1}}}{\in \mathbb{R}^{\mathrm{#1 \times #2}}}}
\newcommand{\Cset}[2]{\ifthenelse{\equal{#2}{1}}{\in \mathbb{C}^{\mathrm{#1}}}{\in \mathbb{C}^{\mathrm{#1 \times #2}}}}
\newcommand{\vect}[1]{\bm{\MakeLowercase{#1}}}
\newcommand{\mtrx}[1]{\bm{\MakeUppercase{#1}}}
\newcommand{\norm}[2]{\|#1\|_{#2}}

\newcommand{\htransp}[1]{#1^{\mathsf{H}}}
\DeclareMathOperator{\sign}{sign} 
\DeclareMathOperator*{\argmin}{arg\,min\,}

\DeclareMathOperator*{\minim}{minimize\,}
\DeclareMathOperator{\Mr}{\mtrx{M}_r}
\DeclareMathOperator{\Mcn}{\mtrx{M}_c^-}
\DeclareMathOperator{\Mcp}{\mtrx{M}_c^+}

\begin{document}

\mainmatter  % start of an individual contribution

% first the title is needed
\title{Sparsity and cosparsity for audio declipping: a flexible non-convex approach}

% a short form should be given in case it is too long for the running head
\titlerunning{Sparse and cosparse regularizations for audio declipping}

% the name(s) of the author(s) follow(s) next
%
% NB: Chinese authors should write their first names(s) in front of
% their surnames. This ensures that the names appear correctly in
% the running heads and the author index.
%

\author{Sr\dj an Kiti\'c \and Nancy Bertin \and R\'emi Gribonval \thanks{This work was supported in part by the European Research Council, PLEASE project (ERC-StG-2011-277906).}}

\authorrunning{Sparse and cosparse regularizations for audio declipping}
% (feature abused for this document to repeat the title also on left hand pages)

% the affiliations are given next; don't give your e-mail address
% unless you accept that it will be published
\institute{Inria/IRISA, Rennes, PANAMA team\\
\mailsa, \mailsb, \mailsc\\
\url{team.inria.fr/panama}}

%
% NB: a more complex sample for affiliations and the mapping to the
% corresponding authors can be found in the file "llncs.dem"
% (search for the string "\mainmatter" where a contribution starts).
% "llncs.dem" accompanies the document class "llncs.cls".
%

\toctitle{Lecture Notes in Computer Science}
\tocauthor{Authors' Instructions}
\maketitle

\begin{abstract}
This work investigates the empirical performance of the sparse synthesis versus sparse analysis regularization for the ill-posed inverse problem of audio declipping. We develop a versatile non-convex heuristics which can be readily used with both data models. Based on this algorithm, we report that, in most cases,
%outside the nominally equivalent setting, 
the two models perform almost similarly in terms of signal enhancement. However, the analysis version is shown to be amenable for real time audio processing, when certain analysis operators are considered. Both versions outperform state-of-the-art methods in the field, especially for the severely saturated signals.

\keywords{Clipping, audio, sparse, cosparse, non-convex, real-time}
\end{abstract}

\section{Introduction}\label{sec:introduction}

Clipping, or magnitude saturation, is a well-known problem in signal processing, from audio \cite{adler2012audio,kahrs1998applications} to image processing \cite{aydin2008dynamic,naik2003hue} and digital communications \cite{li1997effects}. The focus of this work is audio declipping, to restore clipped audio signals. % affected by clipping.
Audio signals become saturated usually during acquisition, reproduction or A/D conversion. The perceptual manifestation of clipped audio depends on the level of clipping degradation and the audio content. In case of mild to moderate clipping, the listener may notice occasional ``clicks and pops'' during playback. When clipping becomes severe, the audio content is usually perceived as if it was contaminated with a high level of additive noise, which may be explained by the introduction of a large number of harmonics caused by the discontinuities in the degraded signal. In addition to audible artifacts, some recent studies have shown that clipping has a negative impact on Automatic Speech Recognition (ASR) performance \cite{tachioka2014speech,harvilla2014least}.

%\paragraph{Notation}\label{sec:notation}
%
%In the following text, the uppercase italic font is used for functionals (\emph{e.g.} $F$). We use the lowercase boldface font for vectors (\emph{e.g.} $\vect{v}$) and  the uppercase boldface font for matrices (\emph{e.g.} $\mtrx{M}$). The $\mathrm{i}^{\text{th}}$ element of a vector $\vect{v}$ is denoted by $\vect{v}_{\mathrm{i}}$. By using the superscript notation $\vect{v}^{(\mathrm{i})}$, we denote the $\mathrm{i}^{\text{th}}$ iterate of a variable $\vect{v}$ of an algorithm. The lowercase regular font (\emph{e.g.} $\mathrm{n}$) is used for the integer-valued scalar variables. We use the uppercase greek alphabet to denote sets (\emph{e.g.} $\Omega$), with the exception of the sets of complex and real numbers which are denoted by $\mathbb{C}$ and $\mathbb{R}$. The lowercase greek symbols denote scalar constants (\emph{e.g.} $\tau$).

%An illustrative example of a clipped one-dimensional signal is given in figure~\ref{figClipping}. 

In the following text, a sampled audio signal is represented by the vector $\vect{x} \Rset{n}{1}$ and its clipped version is denoted by $\vect{y} \Rset{n}{1}$. The latter can be easily deduced from $\vect{x}$ through the following nonlinear observation model, called \textit{hard clipping}:
\begin{equation}\label{eqClipping}
	\vect{y}_{\mathrm{i}} = \begin{cases}
			\vect{x}_{\mathrm{i}} & \text{for } |\vect{x}_{\mathrm{i}}| \leq \tau,\\
			\sign(\vect{x}_{\mathrm{i}}) \tau & \text{otherwise.}
		\end{cases}
\end{equation}
While idealized, this clipping model is a convenient approximation allowing to clearly distinguish the clipped parts of a signal  by identifying the samples having the highest absolute magnitude. Indices corresponding to ``reliable'' samples of $\vect{y}$ (not affected by clipping) are indexed by $\Omega_r$, while $\Omega_c^+$ and $\Omega_c^-$ index the clipped samples with positive and negative magnitude, respectively.
%\begin{figure}
%	\centering
%	\includegraphics[scale=0.2]{explain_sat.pdf}
%	\caption{Hard clipping example}
%	\label{figClipping}
%\end{figure}

Our goal is to estimate the original signal $\vect{x}$ from its clipped version $\vect{y}$, \emph{i.e.} to ``declip'' the signal $\vect{y}$. Ideally, the estimated signal $\hat{\vect{x}}$ should satisfy natural magnitude constraints in order to be consistent with the clipped observations. Thus, we seek an estimate $\hat{\vect{x}}$ which fulfills the following criteria:
%\begin{align}\label{eqConsistency}
%	\Mr \hat{\vect{x}} & = \Mr \vect{y} \nonumber \\
%	\Mcp \hat{\vect{x}} & \geq \Mcp \vect{y} \\
%	\Mcn \hat{\vect{x}} & \leq \Mcn \vect{y} \nonumber,
%\end{align}
\begin{align}\label{eqConsistency}
	\Mr \hat{\vect{x}} & = \Mr \vect{y} &
	\Mcp \hat{\vect{x}} & \geq \Mcp \vect{y} &
	\Mcn \hat{\vect{x}} & \leq \Mcn \vect{y},
\end{align}
where the matrices $\Mr$, $\Mcn$ and $\Mcp$ are \textit{restriction operators}. These are simply row-reduced identity matrices used to extract the vector elements indexed by the sets $\Omega_r$, $\Omega_c^+$ and $\Omega_c^-$, respectively. We write the constraints~\eqref{eqConsistency} as $\hat{\vect{x}} \in \Gamma(\vect{y})$.

Obviously, consistency alone is not sufficient to ensure uniqueness of $\hat{\vect{x}}$, thus one needs to further regularize the inverse problem. The declipping inverse problem is amenable to several regularization approaches proposed in the literature, such as based on linear prediction \cite{janssen1986adaptive}, minimization of the energy of high order derivatives \cite{harvilla2014least}, psychoacoustics \cite{defraene2013declipping}, sparsity \cite{adler2012audio,kitic2013consistent,siedenburg2014audio,defraene2013declipping,weinstein2011recovering} and cosparsity \cite{kiticaudio} (where we introduced a simplified version of the analysis-based algorithm presented in this paper). The last two \textit{priors}, briefly explained in the next section, enable some state-of-the-art methods in clipping restoration. 

In this paper we empirically compare the performance of the two priors, by means of a declipping algorithm which is easily adaptable to both cases. Our findings are that the sparsity-based version of the algorithm marginally outperforms the cosparsity-based one, but this fact may be attributed to the choice of the stopping criterion. On the other hand, for a class of analysis operators, the cosparsity-based algorithm has very low complexity per iteration, which makes it suitable for real-time audio processing.

\section{The sparse synthesis and sparse analysis data models}\label{secDataModels}

It is well-known that the energy of audio signals is often concentrated either in a small number of frequency components, or in short temporal bursts \cite{plumbley2010sparse}, \emph{i.e.} they are (approximately) time-frequency sparse. 
The traditional sparse synthesis viewpoint \cite{eldar2012compressed,foucart2013mathematical} on this property is that audio signals are well approximated by linearly combining few columns of a \textit{dictionary} matrix $\mtrx{D} \Cset{n}{d}$, $\mathrm{d}\geq{\mathrm{n}}$ such as a Gabor dictionary, \emph{i.e.} $\vect{x} \approx \mtrx{D} \vect{z}$, where $\vect{z} \Cset{d}{1}$ is sparse. A less explored alternative is the cosparse analysis perspective \cite{nam2013cosparse} asserting that $\mtrx{A}\vect{x}$ is approximately sparse, with $\mtrx{A} \Cset{p}{n}$, $\mathrm{p \geq n}$ and \textit{analysis operator}. The two data models are different \cite{elad2007analysis,nam2013cosparse}, unless $\mathrm{p = n}$ and $\mtrx{A} = \mtrx{D}^{-1}$. Finding the sparsest (in the sense of synthesis or analysis) vector $\vect{x}$ satisfying  constraints such as \eqref{eqConsistency} is in general intractable, but convex or greedy heuristics provide efficient algorithms with certain performance guarantees \cite{eldar2012compressed,foucart2013mathematical,nam2013cosparse}.

%Rather than finding the most suitable transform domains, our goal is to explore how each data model performs on the same domain, particularly, in the overcomplete case. After choosing, \emph{e.g.} a dictionary $\mtrx{D} \Cset{n}{d}$, the corresponding analysis operator will be $\mtrx{A} = \htransp{\mtrx{D}} \Cset{d}{n}$, where $\htransp{}$ denotes the Hermitian transpose.

%The data fidelity term $F_d$ should, as expected, encourage the linear constraints presented in \eqref{eqConsistency}. 

\section{Algorithms}\label{secAlgorithm}

%As mentioned before, the problems  \eqref{eqSparse} and \eqref{eqCosparse} cannot be exactly solved by polynomial time algorithms. 
Some empirical evidence \cite{defraene2013declipping,weinstein2011recovering} suggests that standard $\ell_1$ convex relaxation does not perform well for sparse synthesis regularization of the declipping inverse problem. Therefore, we developed an algorithmic framework based on non-convex heuristics, that can be straightforwardly parametrized for use in both the synthesis and the analysis setting. To allow for possible real-time implementation, the algorithms operate on individual blocks (chunks) of audio data, which is subsequently resynthesized by means of the overlap-add scheme. 

The heuristics should approximate the solution of the following synthesis- and analysis-regularized inverse problems\footnote{Observe that if $\mtrx{D}$ and $\mtrx{A}$ are unitary matrices, the two problems become identical.}:
\begin{align}
	\minim_{\vect{x}, \vect{z}}  & \norm{\vect{z}}{0} + \mathbf{1}_{\Gamma(\vect{y})}(\vect{x}) + \mathbf{1}_{\ell_{2} \leq \varepsilon}(\vect{x} - \mtrx{D}\vect{z}) \label{eqNPsparse} \\
	\minim_{\vect{x}, \vect{z}}  & \norm{\vect{z}}{0} + \mathbf{1}_{\Gamma(\vect{y})}(\vect{x}) + \mathbf{1}_{\ell_{2} \leq \varepsilon}(\mtrx{A}\vect{x} - \vect{z}) \label{eqNPcosparse}.	
\end{align}
%The data fidelity functional $F_d=F_{\Gamma} + F_c$ consists of $F_{\Gamma} = \mathbf{1}_{\Gamma}$, which 
The indicator function $\mathbf{1}_{\Gamma(\vect{y})}$ of the constraint set $\Gamma(\vect{y})$ forces the estimate $\vect{x}$ to satisfy \eqref{eqConsistency}. The additional penalty $\mathbf{1}_{\ell_{2} \leq \varepsilon}$ is a \textit{coupling} functional. Its role is to enable the end-user to explicitly bound the distance between the estimate and its sparse approximation. These are difficult optimization problems: besides inherited NP-hardness, the two problems are also non-convex and non-smooth.

%Our approach is to first temporarily set $\varepsilon = 0$. 
%%substitute the $\mathbf{1}_{\ell_{2} \leq \varepsilon}$ constraint by a linear equality constraint, 
%%leading to:
%%\begin{align}
%%	\minim_{\vect{x}, \vect{z}}  & %\norm{\vect{z}}{0} + \mathbf{1}_{\Gamma}(\vect{x}) + \tilde{F}_c(\vect{z},\vect{x}) = 
%%	\norm{\vect{z}}{0} + \mathbf{1}_{\Gamma}(\vect{x}) \text{ s. t. } \vect{x} = \mtrx{D}\vect{z}\label{eqNPsparseSmooth} \\
%%	\minim_{\vect{x}, \vect{z}}  & %\norm{\vect{z}}{0} + \mathbf{1}_{\Gamma}(\vect{x})  + \tilde{F}_c(\vect{z},\vect{x}) =  
%%	\norm{\vect{z}}{0} + \mathbf{1}_{\Gamma}(\vect{x}) \text{ s. t. }\mtrx{A}\vect{x} = \vect{z}\label{eqNPcosparseSmooth}.	
%%\end{align}
%For natural signals, we can expect the optimum $\vect{z}^{*}$ to satisfy $\norm{\vect{z}^*}{0} = \mathrm{d}$, which is not our goal. Instead, we intend to roughly approximate these problems by a first order optimization method --such that the equality constraints $\vect{x} = \mtrx{D}\vect{z}$ (resp. $\mtrx{A}\vect{x} = \vect{z}$) are not fully satisfied)-- hopefully yielding the solution which is a feasible point of \eqref{eqNPsparse} or \eqref{eqNPcosparse}. 

%Note that w
We can represent \eqref{eqNPsparse} and \eqref{eqNPcosparse} in an equivalent form, using the indicator function on the cardinality of $\vect{z}$ and an integer-valued unknown $\mathrm{k}$:
\begin{equation}
	\minim_{\vect{x}, \vect{z}, \mathrm{k}}   \mathbf{1}_{\ell_0 \leq \mathrm{k}} (\vect{z}) +  \mathbf{1}_{\Gamma(\vect{y})}(\vect{x}) + F_c(\vect{x},\vect{z}) 
	\label{eq:NPequiv}
\end{equation}
where $F_{c}(\vect{x},\vect{z})$ is the appropriate coupling functional. For a fixed $\mathrm{k}$, problem~\eqref{eq:NPequiv} can be seen as a variant of the \textit{regressor selection} problem, which is (locally) solvable by the Alternating Direction Method of Multipliers (ADMM) \cite{boyd2011distributed,bertsekas1999nonlinear}: %. The ADMM iterates are defined as follows:
\begin{align}\label{eqRegressorSelection}
\begin{aligned}
	& \textrm{Synthesis version}\\
	\bar{\vect{z}}^{\mathrm{(i+1)}} = & \mathcal{H}_{\mathrm{k}}(\hat{\vect{z}}^{\mathrm{(i)}} + \vect{u}^{\mathrm{(i)}})\\
	\hat{\vect{z}}^{\mathrm{(i+1)}} = & \argmin_{\vect{z}} \norm{\vect{z} - \bar{\vect{z}}^{(\mathrm{i+1})} + \vect{u}^{\mathrm{(i)}}}{2}^2  \\
	& \text{subject to } \mtrx{D}\vect{z} \in \Gamma(\vect{y}) \\
	\vect{u}^{\mathrm{(i+1)}} = & \vect{u}^{\mathrm{(i)}} + \hat{\vect{z}}^{\mathrm{(i+1)}} - \bar{\vect{z}}^{\mathrm{(i+1)}}
\end{aligned}
\quad \vline \quad
\begin{aligned}
	& \textrm{Analysis version}\\
	\bar{\vect{z}}^{\mathrm{(i+1)}} = & \mathcal{H}_{\mathrm{k}}(\mtrx{A} \hat{\vect{x}}^{\mathrm{(i)}} + \vect{u}^{\mathrm{(i)}})\\
	\hat{\vect{x}}^{\mathrm{(i+1)}} = & \argmin_{\vect{x}} \norm{\mtrx{A}\vect{x} - \bar{\vect{z}}^{(\mathrm{i+1})} + \vect{u}^{\mathrm{(i)}}}{2}^2  \\
	& \text{subject to } \vect{x} \in \Gamma(\vect{y}) \\
	\vect{u}^{\mathrm{(i+1)}} = & \vect{u}^{\mathrm{(i)}} + \mtrx{A} \hat{\vect{x}}^{\mathrm{(i+1)}} - \bar{\vect{z}}^{\mathrm{(i+1)}}.
\end{aligned}
\end{align}
The operator $\mathcal{H}_{\mathrm{k}}(\vect{v})$ performs hard thresholding, \emph{i.e.} sets all but $\mathrm{k}$ highest in magnitude components of $\vect{v}$ to zero. Unlike the standard regressor selection algorithm, for which the ADMM multiplier \cite{boyd2011distributed} needs to be appropriately chosen to avoid divergence, the above formulation is independent of its value.

In practice, it is difficult to guess the optimal value of $\mathrm{k}$ beforehand. An adaptive estimation strategy is to periodically increase $\mathrm{k}$ (starting from some small value), perform several runs of \eqref{eqRegressorSelection} for a given $\mathrm{k}$ and repeat the procedure until the constraint embodied by $F_c$ is satisfied. This corresponds to \textit{sparsity relaxation}: as $\mathrm{k}$ gets larger, the estimated $\vect{z}$ becomes less sparse.

The proposed algorithm, dubbed \textit{SParse Audio DEclipper (SPADE)}, comes in two flavors. The pseudocodes for the synthesis version (``\textit{S-SPADE}'') and for the analysis version (``\textit{A-SPADE}'') are given in Algorithm \ref{S-SPADE} and Algorithm \ref{A-SPADE}.
\def\algowidth {5.8cm}

\vspace{-2em}

\begin{minipage}[t]{\algowidth}
\vspace{0pt}
\begin{algorithm}[H]
	\caption{S-SPADE}
	\label{S-SPADE}
	\begin{algorithmic}[1]
		\REQUIRE $\mtrx{D}, \vect{y}, \Mr, \Mcp, \Mcn, \mathrm{s}, \mathrm{r}, \varepsilon$
		\STATE $\hat{\vect{z}}^{(0)} = \htransp{\mtrx{D}} \vect{y}, \vect{u}^{(0)} = \vect{0},\mathrm{i}=1, \mathrm{k} = \mathrm{s}$
		\STATE $\bar{\vect{z}}^{(\mathrm{i})} = \mathcal{H}_{\mathrm{k}} \left( \hat{\vect{z}}^{(\mathrm{i}-1)} + \vect{u}^{(\mathrm{i}-1)} \right)$		
		\STATE $\hat{\vect{z}}^{(\mathrm{i})} = \text{arg min}_{\vect{z}} \norm{ \vect{z} - \bar{\vect{z}}^{(\mathrm{i})} + \vect{u}^{(\mathrm{i}-1)} }{2}^2$ 
		\\ s.t. $\vect{x} = \mtrx{D} \vect{z} \in \Gamma$
		\IF{$\norm{\hat{\vect{z}}^{(\mathrm{i})} - \bar{\vect{z}}^{(\mathrm{i})}}{2} \leq \varepsilon$}
			\STATE terminate
		\ELSE
			\STATE $\vect{u}^{(\mathrm{i})} =  \vect{u}^{(\mathrm{i}-1)} + \hat{\vect{z}}^{(\mathrm{i})} - \bar{\vect{z}}^{(\mathrm{i})}$		
			\STATE $\mathrm{i} \gets \mathrm{i} + 1$
			\IF{$i$ \text{mod} $\mathrm{r} = 0$} 
				\STATE $\mathrm{k} \gets \mathrm{k} + \mathrm{s}$
			\ENDIF
			\STATE go to 2
		\ENDIF
		\RETURN $\hat{\vect{x}} = \mtrx{D} \hat{\vect{z}}^{(\mathrm{i})}$
	\end{algorithmic}	
\end{algorithm}
\end{minipage}%
%\begin{minipage}[t]{1.5cm}
%\vspace{10pt}
%\end{minipage}%
\begin{minipage}[t]{\algowidth}
\vspace{0pt}
\begin{algorithm}[H]
	\caption{A-SPADE}
	\label{A-SPADE}
	\begin{algorithmic}[1]
		\REQUIRE $\mtrx{A}, \vect{y}, \Mr, \Mcp, \Mcn, \mathrm{s}, \mathrm{r}, \varepsilon$
		\STATE $\hat{\vect{x}}^{(0)} = \vect{y}, \vect{u}^{(0)} = \vect{0},\mathrm{i}=1, \mathrm{k} = \mathrm{s}$
		\STATE $\bar{\vect{z}}^{(\mathrm{i})} = \mathcal{H}_{\mathrm{k}} \left( \mtrx{A} \hat{\vect{x}}^{(\mathrm{i}-1)} + \vect{u}^{(\mathrm{i}-1)} \right)$		
		\STATE $\hat{\vect{x}}^{(\mathrm{i})} = \text{arg min}_{\vect{x}} \norm{ \mtrx{A} \vect{x} - \bar{\vect{z}}^{(\mathrm{i})} + \vect{u}^{(\mathrm{i}-1)} }{2}^2$\\
		 s.t. $\vect{x} \in \Gamma$
		\IF{$\norm{\mtrx{A} \hat{\vect{x}}^{(\mathrm{i})} - \bar{\vect{z}}^{(\mathrm{i})}}{2} \leq \varepsilon$}
			\STATE terminate
		\ELSE
			\STATE $\vect{u}^{(\mathrm{i})} =  \vect{u}^{(\mathrm{i}-1)} + \mtrx{A} \hat{\vect{x}}^{(\mathrm{i})} - \bar{\vect{z}}^{(\mathrm{i})}$		
			\STATE $\mathrm{i} \gets \mathrm{i} + 1$
			\IF{$i$ \text{mod} $\mathrm{r} = 0$} 
				\STATE $\mathrm{k} \gets \mathrm{k} + \mathrm{s}$
			\ENDIF
			\STATE go to 2
		\ENDIF
		\RETURN $\hat{\vect{x}} = \hat{\vect{x}}^{(\mathrm{i})}$
	\end{algorithmic}	
\end{algorithm}
\end{minipage}

\vspace{1em}

The relaxation rate and the relaxation stepsize are controlled by the integer-valued parameters $\mathrm{r}>0$ and $\mathrm{s}>0$, while the parameter $\varepsilon>0$ is the stopping threshold.

\newtheorem{lmConvergence}{Lemma}
\begin{lmConvergence}
	 The SPADE algorithms terminate in no more than $\mathrm{i}=\lceil \mathrm{d}\mathrm{r}/\mathrm{s} + 1 \rceil$ iterations.
\end{lmConvergence}

\begin{proof}
Once $\mathrm{k} \geq \mathrm{d}$, the hard thresholding operation $\mathcal{H}_{\mathrm{k}}$ becomes an identity mapping. Then, the minimizer of the constrained least squares step $3$ is $\hat{\vect{z}}^{(\mathrm{i-1})}$ (respectively, $\hat{\vect{x}}^{(\mathrm{i-1})}$) and the distance measure in the step $4$ is equal to $\norm{\vect{u}^{(\mathrm{i}-1)}}{2}$. But, in the subsequent iteration, $\vect{u}^{(\mathrm{i}-1)} = \vect{0}$ and the algorithm terminates.
\end{proof}

This bound is quite pessimistic: in practice, we observed that the algorithm terminates much sooner, which suggest that there might be a sharper upper bound on the iteration count.

\section{Computational aspects} %Choice of the dictionary / analysis operator}\label{subsec:choice}

The general form of the \textit{SPADE} algorithms does not impose restrictions on the choice of the dictionary nor the analysis operator. From a practical perspective, however, it is important that the complexity per iteration is kept low. The dominant cost of \textit{SPADE} is in the evaluation of the linearly constrained least squares minimizer step, whose computational complexity can be generally high. Fortunately, for some choices of $\mtrx{D}$ and $\mtrx{A}$ this cost is dramatically reduced.

Namely, if the matrix $\htransp{\mtrx{A}}$ forms a \textit{tight frame} ($\htransp{\mtrx{A}}\mtrx{A} = \zeta \mtrx{I}$), it is easy to show that the step $3$ of A-SPADE reduces to\footnote{Recall that the matrices $\Mr$, $\Mcp$ and $\Mcn$ are tight frames by design.}:
		\begin{align*}
			\vect{x}^{(\mathrm{i})} & = \mathcal{P}_{\Xi} \left( \frac{1}{\zeta} \htransp{\mtrx{A}} (\bar{\vect{z}}^{(\mathrm{i})} - \vect{u}^{(\mathrm{i}-1)}) \right), \text{ where: }\\
			\Xi & = \{ \vect{x} \mid 
			\left[ \begin{smallmatrix} 
				-\Mcp \\
				\Mcn
			\end{smallmatrix} \right] \vect{x} \leq
			\left[ \begin{smallmatrix}
			 -\Mcp  \\
			 \Mcn 
			 \end{smallmatrix} \right]
			 \vect{y} \text{ and }
			\Mr \vect{x} = \Mr \vect{y}  \}.	
		\end{align*}
The projection $\mathcal{P}_{\Xi}(\cdot)$ is straightforward and corresponds to component-wise mappings, thus the per iteration cost of the algorithm is reduced to the cost of evaluating matrix-vector products. 

On the other hand, for S-SPADE this simplification is not possible and the constrained minimization in step 3 needs to be computed iteratively. However, by exploiting the tight frame property of $\mtrx{D}=\htransp{\mtrx{A}}$ and the Woodbury matrix identity, one can build an efficient algorithm that solves this optimization problem with low complexity. 

Finally, the computational cost can be further reduced if the matrix-vector products with $\mtrx{D}$ and $\mtrx{A}$ can be computed with less than quadratic cost. Some transforms that support both tight frame property and fast product computation are also favorable in our audio (co)sparse context. Such well-known transforms are Discrete Fourier Transform, (Modified) Discrete Cosine Transform, (Modified) Discrete Sine Transform and Discrete Wavelet Transform, for instance.

\section{Experiments}

The experiments are aimed to highlight differences in signal enhancement performance between S-SPADE and A-SPADE, and implicitly, the sparse and cosparse data models. It is noteworthy that in the formally equivalent setting ($\mtrx{A} = \mtrx{D}^{-1}$), the two algorithms become identical. As a sanity-check, we include this setting in the experiments. The relaxation parameters are set to $\mathrm{r}=1$ and $\mathrm{s}=1$, and the stopping threshold is $\varepsilon=0.1$.

In addition to \textit{SPADE} algorithms, we also include Consistent IHT \cite{kitic2013consistent} and social sparsity declipping algorithm \cite{siedenburg2014audio} as representatives of state-of-the-art. The latter two algorithms use the sparse synthesis data model for regularizing the declipping inverse problem. Consistent IHT is a low-complexity algorithm based on famous Iterative Hard Thresholding for Compressed Sensing \cite{blumensath2009iterative}, while the social sparsity declipper is based on a structured sparsity prior \cite{kowalski2013social}.

As mentioned before, this work is not aimed towards investigating the appropriateness of various time-frequency transforms in the context of audio recovery, which is why we choose traditional Short Time Fourier Transform (STFT) for all experiments. We use sliding square-rooted Hamming window of size $1024$ samples with $75\%$ overlap. The redundancy level of the involved frames (corresponding to \textit{per-chunk} inverse DFT for the dictionary and forward DFT for the analysis operator) is $1$ (no redundancy), $2$ and $4$. The social sparsity declipper, based on Gabor dictionary, requires batch processing of the whole signal. We adjusted the temporal shift, the window and the number of frequency bins in accordance with previously mentioned STFT settings \footnote{We use the implementation kindly provided by the authors.}. 

For a measure of performance, we use a simple difference between \linebreak signal-to-distortion ratios of clipped ($\text{SDR}_{\vect{y}}$) and processed ($\text{SDR}_{\hat{\vect{x}}}$) signals:
\begin{align*}
\text{SDR}_{\vect{y}} = 20 \log_{10} \frac{\| \bigl[ \begin{smallmatrix} \Mcp\\ \Mcn \end{smallmatrix} \bigr] \vect{x} \|_2}{\| \bigl[ \begin{smallmatrix} \Mcp\\ \Mcn \end{smallmatrix} \bigr] \vect{x} - \bigl[ \begin{smallmatrix} \Mcp\\ \Mcn \end{smallmatrix} \bigr] \vect{y} \|_2},
\text{SDR}_{\hat{\vect{x}} } = 20 \log_{10} \frac{\| \bigl[ \begin{smallmatrix} \Mcp\\ \Mcn \end{smallmatrix} \bigr] \vect{x} \|_2}{\| \bigl[ \begin{smallmatrix} \Mcp\\ \Mcn \end{smallmatrix} \bigr] \vect{x} - \bigl[ \begin{smallmatrix} \Mcp\\ \Mcn \end{smallmatrix} \bigr] \hat{\vect{x}}  \|_2}
\end{align*}
Hence, only the samples corresponding to clipped indices are taken into account. Concerning \textit{SPADE}, this choice makes no difference, since the remainder of the estimate $\hat{\vect{x}}$ perfectly fits the observations $\vect{y}$. However, it may favor the other two algorithms that do not share this feature.

Audio examples consist of 10 music excerpts taken from RWC database \cite{goto2002rwc}, which significantly differ in tonal and vocal content. The excerpts are of approximately similar duration ($\sim10\text{s}$), and are sampled at $16\text{kHz}$ with $16\text{bit}$ encoding. The inputs are generated by artificially clipping the audio excerpts at five levels, ranging from severe ($\text{SDR}_{\vect{y}} = 1\text{dB}$) towards mild ($\text{SDR}_{\vect{y}} = 10\text{dB}$). 

According to the results presented in figure \ref{figResults}, the \textit{SPADE} algorithms yield highest improvement in SDR among the four considered approaches. As assumed, \linebreak S-SPADE and A-SPADE achieve similar results in a non-redundant setting, but when the overcomplete frames are considered, the synthesis version performs somewhat better. Interestingly, the overall best results for the analysis version are obtained for the twice-redundant frame, while the performance slightly drops for the redundancy four. This is probably due to the absolute choice of the parameter $\varepsilon$, and suggests that in the analysis setting, this value should be replaced by a relative threshold instead. In the non-redundant case, declipping by A-SPADE and Consistent IHT took (on the average) $3\text{min}$ and $7\text{min}$, respectively, while the other two algorithms were much slower\footnote{All algorithms were implemented in Matlab\textsuperscript{\circledR}, and run in single-thread mode.} (on the order of hours).

\section{Conclusion}

We presented a novel algorithm for non-convex regularization of the declipping inverse problem. The algorithm is flexible in terms that it can easily accommodate sparse (S-SPADE) or cosparse (A-SPADE) prior, and as such has been used to compare the recovery performance of the two data models. The empirical results are slightly in favor of the sparse synthesis data model. However, the analysis version does not fall far behind, which makes it attractive for practical applications. Indeed, due to the natural way of imposing clipping consistency constraints, it can be implemented in an extremely efficient way, even allowing for a real-time signal processing. Benchmark on real audio data demonstrates that both versions outperform considered state-of-the-art algorithms in the field.

\begin{figure}[H]
	\centering
	\includegraphics[scale=0.25]{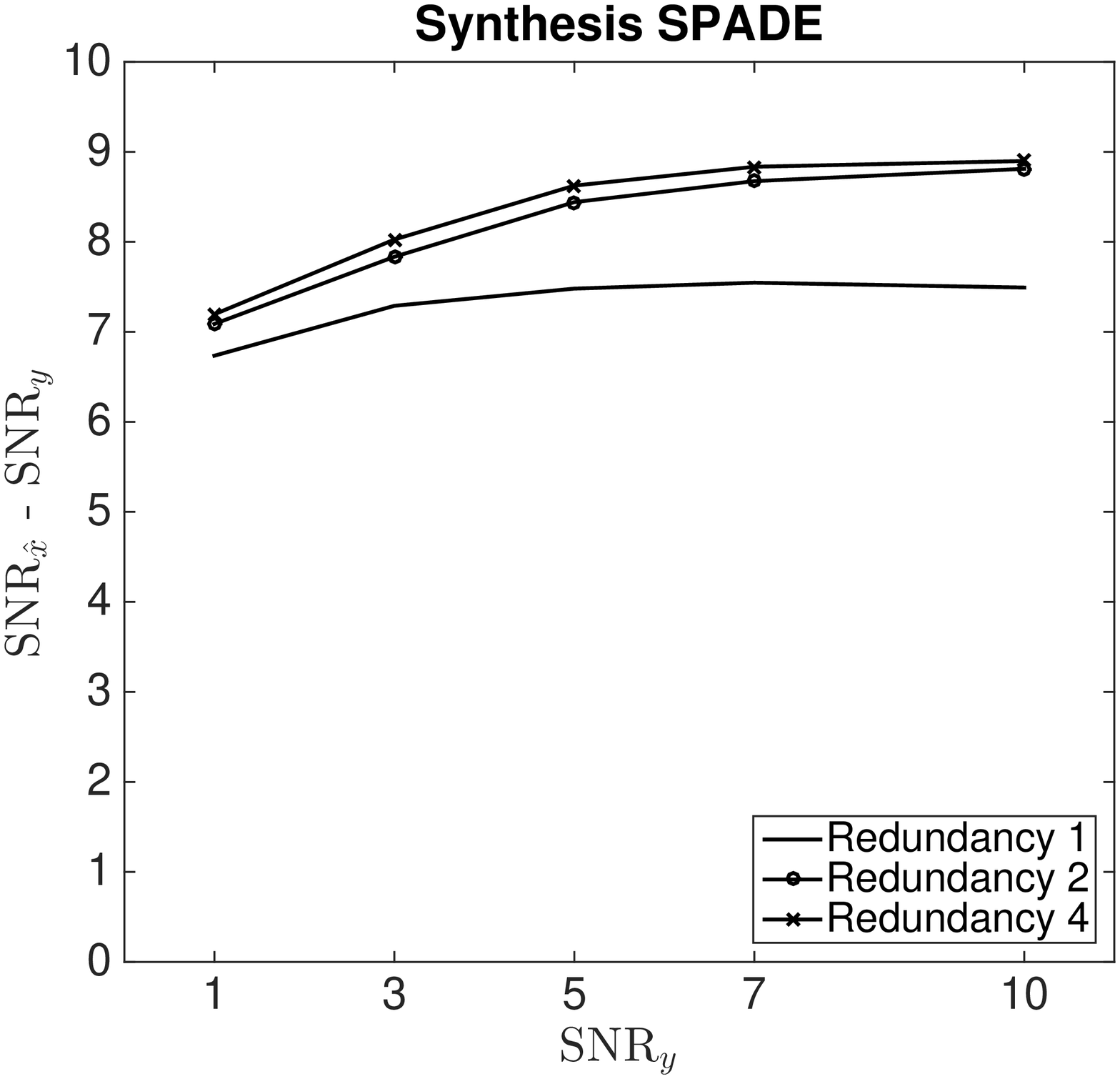}
	\includegraphics[scale=0.25]{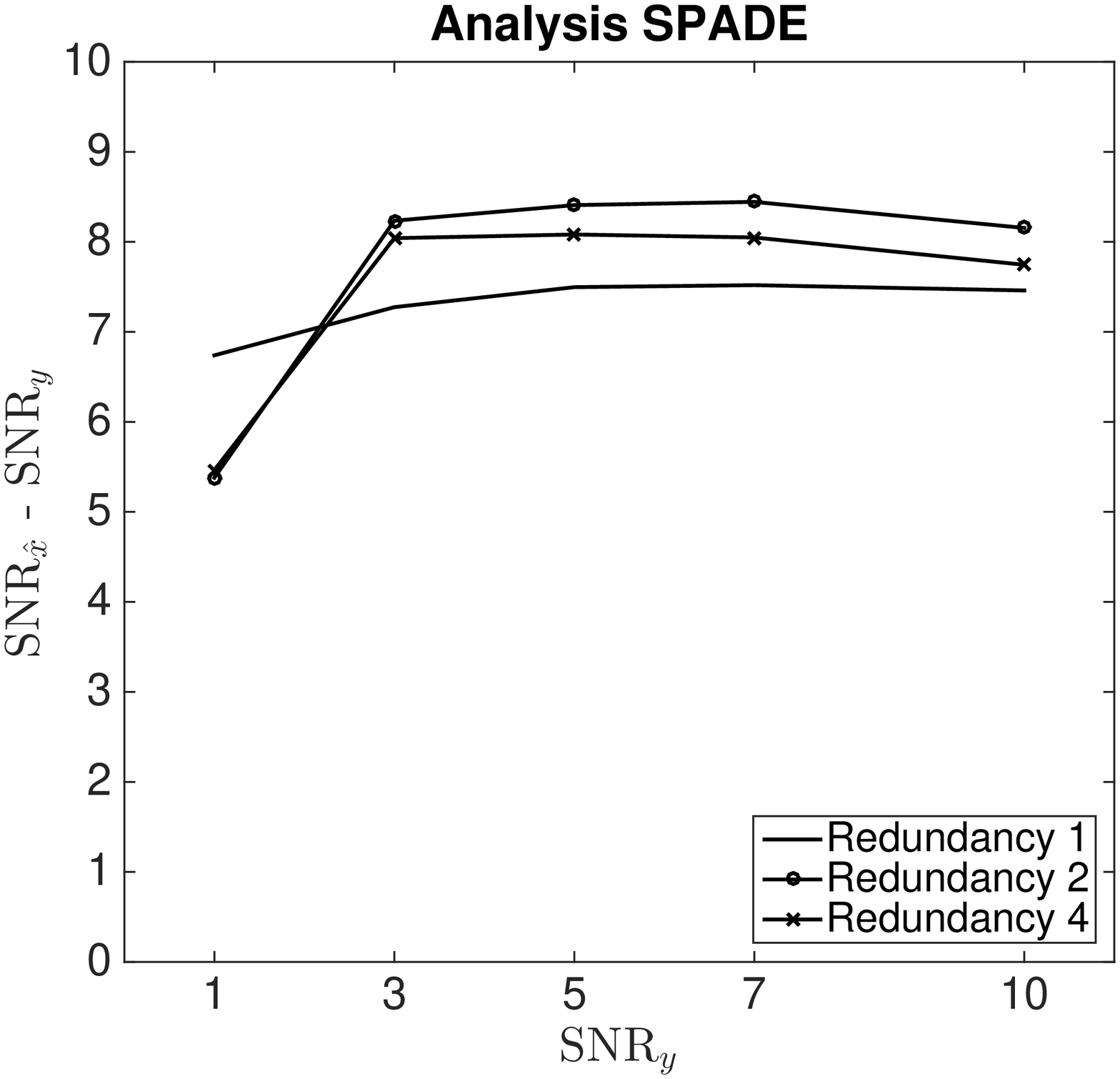}\\	
	\includegraphics[scale=0.25]{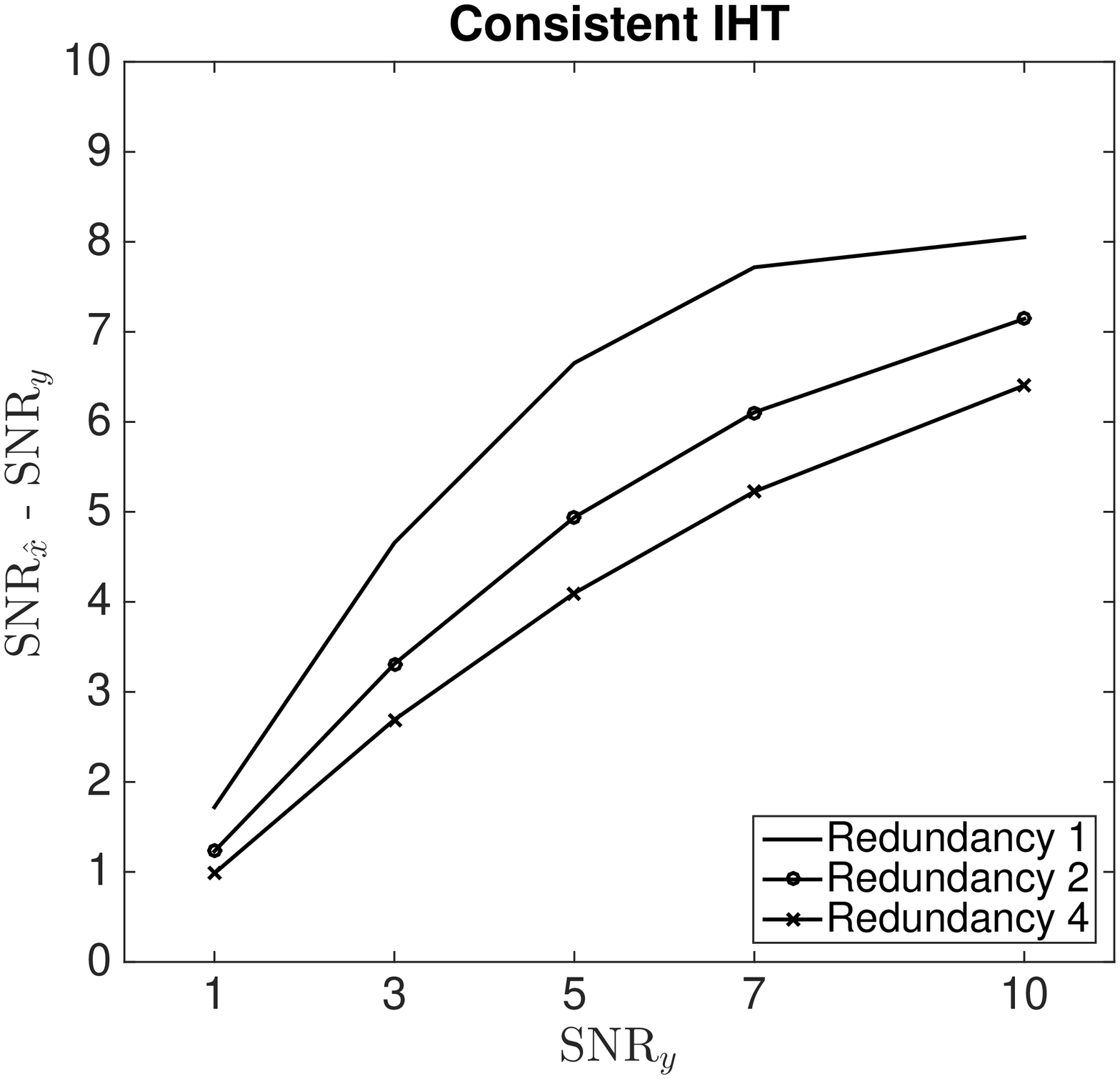}
	\includegraphics[scale=0.25]{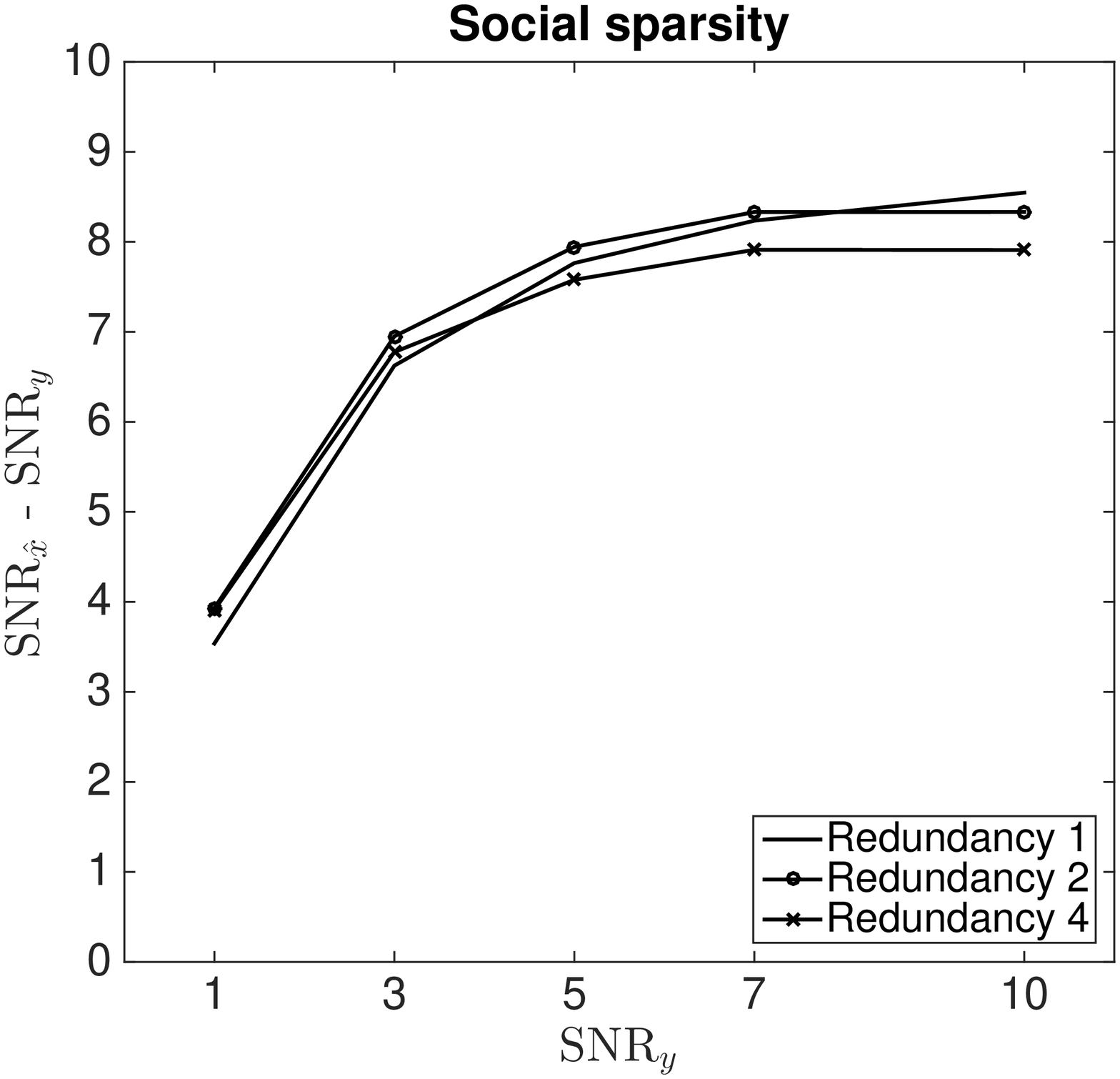}
	\caption{Declipping performance in terms of the $\text{SDR}$ improvement.}\label{figResults}
\end{figure}

Future work will be dedicated to theoretical analysis of the algorithm, with emphasis on convergence. A possible extension is envisioned by introducing structured (co)sparsity priors in the presented algorithmic framework.

\bibliographystyle{abbrv}
\bibliography{Declipping}

\end{document}